\documentclass[a4paper]{article}

\usepackage[top=38truemm,bottom=35truemm,left=35truemm,right=35truemm]{geometry}    

\usepackage{amsmath, amssymb}

\usepackage{amsthm}
\usepackage{mathtools}  
\usepackage{physics}    

\mathtoolsset{showonlyrefs}

\newtheorem{thm}{Theorem}[section]
\newtheorem{prop}[thm]{Proposition} 
\newtheorem{lem}[thm]{Lemma}

\theoremstyle{definition}
\newtheorem{defi}[thm]{Definition}
\newtheorem{notation}{Notation}

\newcommand{\z}{\zeta}

\begin{document}

\title{\textbf{\large PERIODICITY AND ABSOLUTE ZETA FUNCTIONS\\OF MULTI-STATE GROVER WALKS ON CYCLES}
\vspace{15mm}}

\author{Jir\^o AKAHORI$^{1}$, $\quad$ Norio KONNO$^{2, \ast}$, 
\\
\\
Iwao SATO$^{3}$, $\quad$ Yuma TAMURA$^{4}$,
\\ 
\\ 
\\
Department of Mathematical Sciences \\
College of Science and Engineering \\
Ritsumeikan University \\
1-1-1 Noji-higashi, Kusatsu, 525-8577, JAPAN \\
e-mail: akahori@se.ritsumei.ac.jp$^{1}$, \ n-konno@fc.ritsumei.ac.jp$^{2,\ast},$\\
ytamura11029@gmail.com$^{4}$ 
\\
\\
Oyama National College of Technology \\
771 Nakakuki, Oyama 323-0806, JAPAN \\
e-mail: isato@oyama-ct.ac.jp$^{3}$
}

\date{\empty }

\maketitle

\newpage




\begin{small}
\par\noindent
{\bf Corresponding author}: Yuma Tamura, Department of Mathematical Sciences, College of Science and Engineering, Ritsumeikan University, 1-1-1 Noji-higashi, Kusatsu, 525-8577, JAPAN \\
e-mail: ytamura11029@gmail.com

\end{small}

\clearpage

\begin{abstract}
Quantum walks, the quantum counterpart of classical random walks, are extensively studied for their applications in mathematics, quantum physics, and quantum information science. This study explores the periods and absolute zeta functions of Grover walks on cycle graphs. Specifically, we investigate Grover walks with an odd number of states and determine their periods for cycles with any number of vertices greater than or equal to two. In addition, we compute the absolute zeta functions of M-type Grover walks with finite periods. These results advance the understanding of the properties of Grover walks and their connection to absolute zeta functions.
\end{abstract}

\vspace{10mm}

\begin{small}
\end{small}

\vspace{10mm}

\section{Introduction \label{sec01}}
This study builds on the results presented in \cite{Konno2024,AKST2024}. Quantum walks can be seen as the quantum analog of classical random walks and have significant applications in areas such as mathematics, quantum physics, and quantum information science. For detailed discussions on quantum walks, see \cite{GZ, Konno2008, ManouchehriWang, Portugal, Venegas}, and for random walks, refer to \cite{Norris, Spitzer}. On the other hand, absolute zeta functions are defined over $\mathbf{F}_1$, which is sometimes considered a limit of $\mathbf{F}_p$ as $p \to 1$, where $\mathbf{F}_p = \mathbf{Z}/p \mathbf{Z}$ represents the finite field with $p$ elements, $p$ is a prime number. For an introduction to absolute zeta functions, see \cite{CC, KF1, Kurokawa3, Kurokawa, KO, KT3, KT4, Soule}.

Our main results relate to periods of Grover walks. It is a hotly studied topic, and the period of Grover walks on various graphs has been analyzed. \cite{AKST2024, HKSS2017, Yoshie2017, KSTY2018, KKKS2019, KSTY2019, Yoshie2019, IMT2020, Kubota2022, Yoshie2023} Among Grover walks, we examined Grover walks on cycle graphs with an odd number of states and determined their periods on cycles with any number of vertices greater than or equal to two. Furthermore, in the case of M-type Grover walks with finite periods, we computed the absolute zeta functions of those quantum walks.

The structure of this paper is as follows. Section \ref{pre} introduces absolute zeta functions and cyclotomic polynomials, which appears in this study. Section \ref{CycleGraphs} provides a basic overview of multi-state Grover walks on cycle graphs. In Section \ref{mainthms}, we present our main theorems, which clarify the periods of multi-state Grover walks on cycle graphs. Absolute zeta functions of zeta functions of M-type Grover walks also presented in the section.

\section{Preliminaries}\label{pre}

\subsection{Cyclotomic polynomials}\label{CyclotomicPoly}
First we introduce the following notation: $\mathbf{Z}$ is the set of integers, $\mathbf{Z}_{>0} = \{1,2,3, \ldots \}$, $\mathbf{Q}$ is the set of rational numbers, $\mathbf{R}$ is the set of real numbers, and $\mathbf{C}$ is the set of complex numbers. 

In this subsection, we briefly introduce cyclotomic polynomials, because they are relevant to this paper.
\begin{defi}
    $ \mathbf{Z}[x] $ and $ \mathbf{Q}[x] $ denote the polynomial rings with integer and rational coefficients, respectively.
\end{defi}
Then \emph{cyclotomic polynomials} are defined as follows:
\begin{defi}
    For $ n \in \mathbf{Z}_{>0} $, the \emph{\(n\)th cyclotomic polynomial} $ \Phi_n(x) $ is defined by the following formula:
    \begin{align*}
        \quad \Phi_n(x) &:= \prod_{ \substack{ 1 \le k \le n \\ \gcd(k,n) = 1 } } ( x - e^{ 2 k \pi i / n } ).
    \end{align*}
\end{defi}
Note that $ \Phi_n( x ) \in \mathbf{Z}[x] $ for all $n$. Now, the subsequent proposition is the key of this paper.
\begin{prop}[See e.g. \cite{HiguchiEtAl2017}]\label{monic_unity}
    If all of the roots of a monic polynomial with rational coefficients $ f(x) $ are roots of unity, then $ f(x) \in \mathbf{Z}[x] $.
\end{prop}
Here, \emph{monic} means ``the nonzero coefficient of highest degree is equal to $1$,'' and \emph{root of unity} means a complex number $ z $ which satisfies
\[
    z^n = 1
\]
for some $ n \in \mathbf{Z}_{>0} $. This proposition is a consequence of the fact that the minimal polynomial over $ \mathbf{Q} $ of any root of unity is a cyclotomic polynomial, in particular, a polynomial with integer coefficient.

In this paper, this proposition is used in the following manner. First, note that for any square matrix $ A $ and positive integer $n$, if $ \lambda $ is an eigenvalue of $ A $, then $ \lambda^n $ is an eigenvalue of $ A^n $. Therefore, if $ A $ has a complex number which is not root of unity as its eigenvalue, then every $ A^n $ has a complex number which is not equal to $ 1 $ as its eigenvalue. Consequently, we know that $ A^n $ is not the identity matrix for any positive integer \(n\). That is, to prove that $ A^n $ is not the identity matrix for any positive integer \(n\), it is enough to find at least one non-integer coefficient of the eigenpolynomial of \(A\).

\subsection{Absolute zeta functions \label{sec02}}

In this subsection, we briefly review the framework on absolute zeta functions, which can be considered as zeta function over $\mathbf{F}_1$, and absolute automorphic forms (see \cite{Kurokawa3, Kurokawa, KO, KT3, KT4} and references therein, for example). 

Let $f(x)$ be a function $f : \mathbf{R} \to \mathbf{C} \cup \{ \infty \}$. We say that $f$ is an {\em absolute automorphic form} of weight $D$ if $f$ satisfies
\begin{align*}
f \left( \frac{1}{x} \right) = C x^{-D} f(x)
\end{align*}
with $C \in \{ -1, 1 \}$ and $D \in \mathbf{Z}$. The {\em absolute Hurwitz zeta function} $Z_f (w,s)$ is defined by
\begin{align*}
Z_f (w,s) = \frac{1}{\Gamma (w)} \int_{1}^{\infty} f(x) \ x^{-s-1} \left( \log x \right)^{w-1} dx,
\end{align*}
where $\Gamma (x)$ is the gamma function (see \cite{Andrews1999}, for instance). Then, taking $x=e^t$, we see that $Z_f (w,s)$ can be rewritten as the Mellin transform: 
\begin{align*}
Z_f (w,s) = \frac{1}{\Gamma (w)} \int_{0}^{\infty} f(e^t) \ e^{-st} \ t^{w-1} dt.
\end{align*}
Moreover, the {\em absolute zeta function} $\zeta_f (s)$ is defined by 
\begin{align*}
\zeta_f (s) = \exp \left( \frac{\partial}{\partial w} Z_f (w,s) \Big|_{w=0} \right).
\end{align*}

Here we introduce the {\em multiple Hurwitz zeta function of order $r$}, $\zeta_r (s, x, (\omega_1, \ldots, \omega_r))$, the {\em multiple gamma function of order $r$}, $\Gamma_r (x, (\omega_1, \ldots, \omega_r))$, and the {\em multiple sine function of order $r$}, $S_r (x, (\omega_1, \ldots, \omega_r))$, respectively (see \cite{Kurokawa3, Kurokawa, KT3}, for example): 
\begin{align*}
\zeta_r (s, x, (\omega_1, \ldots, \omega_r))
= \; & \sum_{n_1=0}^{\infty} \cdots \sum_{n_r=0}^{\infty} \left( n_1 \omega_1 + \cdots + n_r \omega_r + x \right)^{-s}, 
\\
\Gamma_r (x, (\omega_1, \ldots, \omega_r)) 
= \; & \exp \left( \frac{\partial}{\partial s} \zeta_r (s, x, (\omega_1, \ldots, \omega_r)) \Big|_{s=0} \right),
\\
S_r (x, (\omega_1, \ldots, \omega_r))
= \; & \Gamma_r (x, (\omega_1, \ldots, \omega_r))^{-1} \ \Gamma_r (\omega_1+ \cdots + \omega_r - x, (\omega_1, \ldots, \omega_r))^{(-1)^r}.
\end{align*}
\par
Now we present the following key result derived from Theorem 4.2 and its proof in Korokawa \cite{Kurokawa} (see also Theorem 1 in Kurokawa and Tanaka \cite{KT3}):

\begin{thm}\label{ExplicitAZeta}
    If \(f\) has the form
    \[
        f(x) = x^{l/2} \frac{ ( x^{m(1)}-1 ) \cdots ( x^{m(a)}-1 ) }{ ( x^{n(1)}-1 ) \cdots ( x^{n(b)}-1 ) }
    \]
    for some $ l \in \mathbf{Z} $, $ a,b \in \mathbf{Z}_{>0} $, $ m(i), n(j) \in \mathbf{Z}_{>0}\, (i=1,\dots,a, j=1,\dots,b) $, then the following holds:
    \begin{align*}
        Z_f(w,s) = \; & \sum_{ I \subset \{ 1,\dots,a \} } (-1)^{|I|} \zeta_b( w, s-\deg(f) + m(I), \mbox{\boldmath $n$} ),   \\
        \zeta_f(s) = \; & \prod_{ I \subset \{ 1,\dots,a \} } \Gamma_b( s - \deg(f) + m(I), \mbox{\boldmath $n$} )^{ (-1)^{ |I| } },    \\
        \zeta_f( D-s )^C = \; & \varepsilon_f(s) \zeta_f(s),
    \end{align*}
    where
    \begin{gather*}
        \deg(f) = l/2 + \sum_{i=1}^a m(i) - \sum_{j=1}^b n(j), \quad m(I) = \sum_{ i \in I } m(i), \\
        \mbox{\boldmath $n$} = ( n(1), \dots, n(j) ), \quad D = l + \sum_{i=1}^a m(i) - \sum_{j=1}^b n(j), \\
        C = (-1)^{a-b}, \quad \varepsilon_f= \prod_{ I \subset \{ 1,\dots,a \} } S_b( s - \deg(f) + m(I), \mbox{ \boldmath $n$ } )^{(-1)^{|I|}}.
    \end{gather*}
\end{thm}

\section{Grover walks on cycle graphs}\label{CycleGraphs}

\subsection{Quantum walks with odd number states on cycle graphs}\label{QW}

For $ N \ge 2 $, \emph{an undirected cycle graph with $N$ vertices} is an undirected graph which has $ N $ vertices and each vertex connecting to exactly $2$ edges. We write this graph by $C_N$. Formally, $ C_N $ is defined in the following way:

\begin{defi}
    The set of vertices of $ C_N $ is $ \{ 0, 1, \dots, N-1 \} $ and the set of edges of $ C_N $ is $ \{ \{ k, k+1 \} \mid k=0,\dots,N-1 \} $, where the numbers are identified modulo $N$.
\end{defi}
%
A quantum walk is the time-evolving sequence of states consisting of position and chirality. Formally, a state is a vector which is an element of a tensor product of two Hilbert spaces over $ \mathbf{C} $, $ \mathcal{H}_P $ and $ \mathcal{H}_C  $. Furthermore, $ \mathcal{H} := \mathcal{H}_P \otimes \mathcal{H}_C $. In this paper, $ \mathcal{H}_P $ is a vector space over $ \mathbf{C} $ in which $ \{ \ket{k} \mid x \in V(C_N) \} $ is an orthonormal basis, where $V(C_N)$ is the set of vertices of $C_N$. Here, let $ L \ge 3 $ be an odd number and $ m := (L-1)/2 $, i.e., $ L = 2m+1 $. Then, $ \mathcal{H}_C $ is a vector space over $ \mathbf{C} $ in which $ \{ \ket{\leftarrow m}, \dots, \ket{ \leftarrow 1 }, \ket{\cdot}, \ket{ 1 \rightarrow}, \dots, \ket{ m \rightarrow} \} $ is an orthonormal basis. Therefore, each state can be represented like the following:
\[
    \sum_{ k \in V(C_n) } \ket{k} \otimes s, \qquad s \in \mathcal{H}_C.
\]
Usually, we assume that the initial state $ \Psi_0 $ satisfies $ \| 
\Psi_0 \| = 1 $. Moreover, we consider the case where the time-evolution operator $ U $ is decomposed as $ U = SC $. Here, $ S $ is called \emph{the shift operator} and defined by the following formulas:
\begin{align*}
    S( \ket{k} \otimes \ket{ \leftarrow j } ) &:= \ket{ k-j } \otimes \ket{ \leftarrow j }, && j=1,\dots,m, \\
    S( \ket{k} \otimes \ket{ j \rightarrow }) &:= \ket{ k+j } \otimes \ket{ \rightarrow j }, && j=1,\dots,m,  \\
    S( \ket{k} \otimes \ket{ \cdot }) &:= \ket{ k } \otimes \ket{ \cdot }.
\end{align*}
Furthermore, $ C $ is called \emph{the coin operator} and defined by the following:
\[
    C := \sum_{ k \in V(C_N) } \ketbra{k}{k} \otimes A
\]
for some unitary operator $A$ on $ \mathcal{H}_C $. We call this operator $ A $ \emph{the local coin operator}. In this case, $ S $ and $ C $ are both unitary, and then $ U $ is also unitary. Now, the time-evolution is defined as usual:
\[
    \Psi_{n+1} := U \Psi_n.
\]
Of course, we have $ \Psi_n = U^n \Psi_0 $. We are interested in this time-evolution operator $U$. In each of the subsequent subsections, matrix representations of $ U $ are shown. Moreover, we introduce the period of a quantum walk.
\begin{defi}
    For a quantum walk whose time-evolution operater is $U$, \emph{the period of the quantum walk} is defined as the infimum:
    \[
        \inf \{ n \ge 1 \mid U^n = 1 \}.
    \]
    If the set in the above formula is empty, then the period is defined to be $ \infty $.
\end{defi}
Of course, if $ T $ is the period of a quantum walk, it holds that
\[
    \Psi_T = \Psi_0
\]
for any initial state $ \Psi_0 $.

Also, we define the zeta function of a quantum walk on a cycle graph:
\begin{defi}
    For a quantum walk on a cycle graph where the matrix representation of the time-evolution operator is $ U $, \emph{the zeta function of the quantum walk} $ \zeta $ is defined as follows:
    \begin{equation*}
        \zeta(u) := \det( I - u U )^{-1},
    \end{equation*}
    where $ I $ is the identity matrix.
\end{defi}
This definition can be seen in \cite{Konno2024} for example.

\subsection{Grover walks}\label{GroverWalks}

Grover walks are a well-studied class of quantum walks. There are two types of Grover walks: M- and F-type. They are characterized by local coin operators as usual. The following is the definition.

\begin{defi}
    \emph{M-type} (respectively, \emph{F-type}) \emph{Grover walks with $ L \ (=2m+1) $ states on $ C_N $} are quantum walks whose local coin operator's matrix representation with respect to the ordered basis $ ( \ket{\leftarrow m}, \dots, \ket{ \leftarrow 1 }, \ket{\cdot}, \ket{ 1 \rightarrow}, \dots, \ket{ m \rightarrow} ) $ is as follows respectively:
    \begin{align*}
        A^{M,L} = \; & \frac{1}{ L }
        \begin{bmatrix}
            -(2m-1)   &   2  & \cdots &   2   \\
             2   &  -(2m-1)  & &  2   \\
             \vdots & & \ddots & \vdots \\
             2   &   2  & \cdots & -(2m-1)
        \end{bmatrix},
        \\
        A^{F,L} = \; & \frac{1}{ L }
        \begin{bmatrix}
             2  & \cdots  &   2  &  -(2m-1)   \\
             2   & \cdots & - (2m -1)  &   2   \\
             \vdots & & \vdots & \vdots  \\
            -(2m-1) & \cdots &   2  &   2
        \end{bmatrix}.
    \end{align*}
\end{defi}

We write $ U^{M,L}_N $ and $ U^{F,L}_N $ for the time-evolution opeartors of M- and F-type Grover walks on $ C_N $, respectively. They are determined by the way in subsection \ref{QW}. For example, the matrix representation of $ U^{M,3}_5 $ with respect to the ordered basis of $ \mathcal{H} $ $ ( \ket{0} \otimes \ket{ \leftarrow 1 },\allowbreak \ket{0} \otimes \ket{\cdot}, \ket{0} \otimes \ket{ \rightarrow 1 }, \ket{1} \otimes \ket{ \leftarrow 1 }, \ket{1} \otimes \ket{\cdot}, \dots, \ket{4} \otimes \ket{ \rightarrow 1 } ) $ is as follows:
\[
    U^{M,3}_5 =
    \begin{bmatrix}
        S       & L & O   & O       & R \\
        R & S       & L &  O       & O       \\
        O       & R & S       &  L       & O       \\
        O       & O       & R &  S       & L \\
        L & O       & O   & R & S
    \end{bmatrix},
\]
where
\begin{align*}
    L :=
    \begin{bmatrix}
        1 & 0 & 0   \\
        0 & 0 & 0   \\
        0 & 0 & 0
    \end{bmatrix}
    A^{M,3},
    &&
    S :=
    \begin{bmatrix}
        0 & 0 & 0   \\
        0 & 1 & 0   \\
        0 & 0 & 0
    \end{bmatrix}
    A^{M,3},
    &&
    R :=
    \begin{bmatrix}
        0 & 0 & 0   \\
        0 & 0 & 0   \\
        0 & 0 & 1
    \end{bmatrix}
    A^{M,3}.
\end{align*}
 As another example, the matrix representation of $ U^{M,5}_4 $ with respect to the ordered basis of $ \mathcal{H} $ $ ( \ket{0} \otimes \ket{ \leftarrow 2 }, \ket{0} \otimes \ket{ \leftarrow 1 }, \ket{0} \otimes \ket{\cdot}, \ket{0} \otimes \ket{ \rightarrow 1 }, \ket{0} \otimes \ket{ \rightarrow 2 }, \ket{1} \otimes \ket{ \leftarrow 2 }, \ket{1} \otimes \ket{ \leftarrow 1 }, \ket{1} \otimes \ket{\cdot}, \dots, \ket{4} \otimes \ket{ \rightarrow 2 } ) $ is as follows:
\[
    U^{M,5}_4 =
    \begin{bmatrix}
        S       & L_1 & L_2+R_2   & R_1 \\
        R_1 & S    & L_1 &  L_2+R_2       \\
        L_2+R_2       & R_1 & S       &  L_1   \\
        L_1       & L_2+R_2       & R_1 &  S
    \end{bmatrix},
\]
where
\begin{gather*}
    L_2 :=
    \begin{bmatrix}
        1 & 0 & 0 & 0 & 0   \\
        0 & 0 & 0 & 0 & 0   \\
        0 & 0 & 0 & 0 & 0   \\
        0 & 0 & 0 & 0 & 0   \\
        0 & 0 & 0 & 0 & 0
    \end{bmatrix}
    A^{M,5},\qquad
    L_1 :=
    \begin{bmatrix}
        0 & 0 & 0 & 0 & 0   \\
        0 & 1 & 0 & 0 & 0   \\
        0 & 0 & 0 & 0 & 0   \\
        0 & 0 & 0 & 0 & 0   \\
        0 & 0 & 0 & 0 & 0
    \end{bmatrix}
    A^{M,5},
    \\
    S :=
    \begin{bmatrix}
        0 & 0 & 0 & 0 & 0   \\
        0 & 0 & 0 & 0 & 0   \\
        0 & 0 & 1 & 0 & 0   \\
        0 & 0 & 0 & 0 & 0   \\
        0 & 0 & 0 & 0 & 0
    \end{bmatrix}
    A^{M,5},
    \\
    R_1 :=
    \begin{bmatrix}
        0 & 0 & 0 & 0 & 0   \\
        0 & 0 & 0 & 0 & 0   \\
        0 & 0 & 0 & 0 & 0   \\
        0 & 0 & 0 & 1 & 0   \\
        0 & 0 & 0 & 0 & 0
    \end{bmatrix}
    A^{M,5},\qquad
    R_2 :=
    \begin{bmatrix}
        0 & 0 & 0 & 0 & 0   \\
        0 & 0 & 0 & 0 & 0   \\
        0 & 0 & 0 & 0 & 0   \\
        0 & 0 & 0 & 0 & 0   \\
        0 & 0 & 0 & 0 & 1
    \end{bmatrix}
    A^{M,5}.
\end{gather*}

Furthermore, let $ f^{M,L}_N(x) $ and $ f^{F,L}_N(x) $ be the characteristic polynomials of $ U^{M,L}_N $ and $ U^{F,L}_N $, respectively:
\begin{align*}
    f^{M,L}_N(x) := \det( x I_{LN} - U^{M,L}_N ),   &&  f^{F,L}_N(x) := \det( x I_{LN} - U^{F,L}_N ).
\end{align*}
%

\section{Main theorems}\label{mainthms}

\subsection{M-type Grover walks}
First, we introduce the following notation.
\begin{notation}
    For \( L,N \ge 2 \), \( k=0,\dots,L-1 \), let
    \[
        f^{M,L}_{N,k}(x) := \det( x I_L - Z_L^k A^{M,L} ),
    \]
    where \(A^{M,L}\) is defined in subsection \ref{GroverWalks} and \( Z_L \) represents the following \( L \times L \) matrix:
    \[
        \begin{bmatrix}
            \z_N^m & 0 & \cdots & 0 \\
            0 & \z_N^{m-1} & & 0 \\
            \vdots & & \ddots & \vdots \\
            0 & 0 & \cdots & \z_N^{-m}
        \end{bmatrix},
    \]
    where $ \z_N := e^{ 2 \pi i / N } $. That is,
    \footnotesize{
    \begin{multline}\label{Mproduct}
        f^{M,L}_{N,k}(x) =
        \\
        \det \frac{1}{L}
        \begin{bmatrix}
            L x + (2m-1) \z_N^{mk} & - 2 \z_N^{mk} & \cdots & - 2 \z_N^{mk}   \\
            - 2 \z_N^{(m-1)k} & L x + (2m-1) \z_N^{(m-1)k} & & - 2 \z_N^{(m-1)k}   \\
            \vdots &  & \ddots & \vdots \\
            -2 \z_N^{-mk} & -2 \z_N^{-mk} & \cdots & L x + (2m-1) \z_N^{-mk}
        \end{bmatrix}.
    \end{multline}
    }
\end{notation}
Using $ f^{M,L}_{N,k}(x) $, we can factor $ f^{M,L}_N(x) $ as follows:
\[
    f^{M,L}_N(x) = \prod_{k=0}^{N-1} f^{M,L}_{N,k}(x).
\]
A proof can be found in \cite{AKST2024}. Here, the following is our main theorems for M-type Grover walk.

\begin{thm}\label{Mmainthm}
    Let \( L \ge 3 \) be a prime number. Then,
    \[
        T^{M,L}_N =
        \begin{cases}
            2L, & N=L,    \\
            \infty, & N \neq L
        \end{cases}
    \]
    for all \( N \ge 2 \).
\end{thm}
\begin{thm}\label{Mmainthm2}
    Let \( L \ge 3 \) be an odd number. Then $T^{M,L}_L = 2L$ and $ T^{M,L}_N = \infty $ if $ N $ has a factor which is coprime to $ L $. 
\end{thm}
The next theorem is obtained for the absolute zeta function.
\begin{thm}\label{MmainthmAbs}
    Let $ L \ge 3 $ be an odd number. Then, the absolute zeta function of \( \zeta^{M,L}_L \), which is the zeta function of \(L\)-state Grover walk on the cycle with \(L\) vertices, can be represented as the following form:
    \begin{align}
        Z_{\zeta^{M,L}_L} (w, s) 
        = \; & \sum_{ I \subset \{ 1,\dots,L-2 \} } (-1)^{|I|+1} \ \zeta_{2L-2} \left(w, s + L^2  + |I|, \mbox{\boldmath \(n\)} \right),
        \\
        \zeta_{\zeta^{M,L}_L} (s)
        = \; & \prod_{ I \subset \{ 1,\dots,L-2 \} } \Gamma_{2L-2} \left( s + L^2 + |I|, \mbox{\boldmath \(n\)} \right)^{ (-1)^{|I|+1}},
        \\
        \zeta_{\zeta^{M,L}_L} (-L^2-s)^{-1} 
        = \; & e_L(s) \zeta_{\zeta^{M,L}_L} (s),
    \end{align}
    where \( \mbox{\boldmath \(n\)} = (2,\dots,2,L,\dots,L) \) (a sequence consisting of $(L-1)$ $2$'s and $(L-1)$ $L$'s) and
    \begin{equation}
        e_L(s) = \prod_{ I \subset \{ 1,\dots,L-2 \} } S_{2L-2} \left( s + L^2  + |I|, (2,\dots,2,L,\dots,L) \right)^{ (-1)^{|I|+1}}.
    \end{equation}
\end{thm}
%
%
To prove these theorems, we prepare the following lemmas.
\begin{lem}\label{MLLk}
    Let \( L \ge 3 \) be an odd number. Then it holds that
    \[
        f^{M,L}_{L,k}(x) =
        \begin{cases}
            (x-1)(x+1)^{L-1}, & k=0,\\
            x^L-1, & k=1,\dots,L-1.
        \end{cases}
    \]
\end{lem}
\begin{lem}\label{Mdevide}
    Let $ L, N_1, N_2 \ge 2 $. If $ N_1 \mid N_2 $, then $ f^{M,L}_{N_1}(x) \mid f^{M,L}_{N_2}(x) $ as polynomials.
\end{lem}
\begin{lem}\label{Mfactorization}
    Let \( L \ge 3 \) be an odd number. Then the coefficients of $ f^{M,L}_{N,k}(x) $ are partially clarified as follows:
    \begin{align}
        f^{M,L}_{N,k}(x) = \ & x^{L} + \dots -\frac{2m-5}{L} \biggl( \sum_{|j|\le3m-3} b_j \z_N^{jk} \biggr) x^3 \\
        &-\frac{2m-3}{L} \biggl( \sum_{|j|\le2m-1} a_j \z_N^{jk} \biggr) x^2 - \frac{2m-1}{L} \biggl( \sum_{|j| \le m} \z_N^{jk} \biggr) x -1,
    \end{align}
    where $ a_j $ $ ( |j| \le 2m-1 ) $ and $ b_j $ $ ( |j| \le 3m-3 ) $ are some integers.
\end{lem}
\begin{lem}\label{Mcoprime}
    Let  $ L $ be an odd number. If $ n $ is coprime to $ L $, then $ f^{M,L}_{n}(x) $ is not an element of \( \mathbf{Z}[x] \).
\end{lem}
\begin{lem}\label{M_L^2}
    Let \( L \) be an odd prime. Then $ f^{M,L}_{L^2}(x) $ is not an element of \( \mathbf{Z}[x] \).
\end{lem}
Once this lemmas are proved, Theorems \ref{Mmainthm} and \ref{Mmainthm2} can be proved as follows. Note that this method of proof can be found in \cite{KKKS2019,AKST2024}.
\begin{proof}[Proof of Theorem \ref{Mmainthm}]
    First, let $ N=L $. Then, by Lemma \ref{MLLk}, $ f^{M,L}_{L}(x) $ has \( -1 \) and \( \z_L \) as its roots and the set of all roots of $ f^{M,L}_{L}(x) $ is included by $ \{ -1 \} \cup \{ \z_L^k \in \mathbf{C} \mid 0, \dots, L-1 \} $. This implies that \( T^{M,L}_L = 2L \).

    Next, let $ N \neq L $. First, consider the case where $ N $ has a factor $n$ which is coprime to $L$. Then by Lemma \ref{Mdevide}, we know
    \[
        f^{M,L}_{n}(x) \mid f^{M,L}_{N}(x).
    \]
    Moreover, $ f^{M,L}_{n}(x) $ is monic and $ f^{M,L}_{n}(x) \in \mathbf{Q}[x] $ by definition, and $ f^{M,L}_{n}(x) \notin \mathbf{Z}[x] $ by Lemma \ref{Mcoprime}. Thus, by Proposition \ref{monic_unity}, $ f^{M,L}_{n}(x) $ has a root which is not a root of unity. Furthermore, the root is also a root of $ f^{M,L}_{N}(x) $. This means $ T^{M,L}_N = \infty $.

    The remainder is the case where \(N\) is a power of $L$ i.e. $N=L^e$ for some $ e \ge 2 $. Since $ f^{M,L}_{L^2}(x) \mid f^{M,L}_{L^e}(x) $, we have to show that $ f^{M,L}_{L^2}(x) $ has a root which is not a root of unity. Now, by the definition of $ f^{M,L}_{L^2}(x) $, Proposition \ref{monic_unity}, and Lemma \ref{M_L^2} imply that $ f^{M,L}_{L^2}(x) $ has a root which is not a root of unity. This completes the proof.
\end{proof}
Theorem \ref{Mmainthm2} can be proved in the same way as the proof above. Furthermore, Theorem \ref{MmainthmAbs} follows directly from Theorem \ref{ExplicitAZeta} and the definition of the absolute zeta functions. Note that, by Lemma \ref{MLLk} and the definition of the zeta functions of quantum walks, we have
\[
    \zeta^{M,L}_L(s) = -\frac{ (s-1)^{L-2} }{ (s^2-1)^{L-1} (s^L-1)^{L-1} }.
\]
 
Also, we will provide proofs of Lemmas \ref{MLLk} to \ref{M_L^2} one by one.
\begin{proof}[Proof of Lemma \ref{MLLk}]
    For the case of $k=0$, we can calculate $ f^{M,L}_{N,0}(x) $ directly by simple way:
    \begin{align}
        f^{M,L}_{N,0}(x) = \; & \frac{1}{L^{2m+1}}
        \begin{vmatrix}
            Lx + (2m-1) & - 2 & \cdots & - 2   \\
            - 2 & Lx + (2m-1) & & - 2   \\
            \vdots &  & \ddots & \vdots \\
            -2 & -2  & \cdots & Lx + (2m-1)
        \end{vmatrix}
        \\
        = \; & \frac{1}{L^{2m+1}}
        \begin{vmatrix}
            Lx-(2m+1) & Lx-(2m+1) & \cdots & Lx - (2m+1)   \\
            - 2 & Lx+(2m-1) & & - 2   \\
            \vdots &  & \ddots & \vdots \\
            -2 & -2  & \cdots & Lx+(2m-1)
        \end{vmatrix}
        \\
        = \; & \frac{L(x-1)}{L^{2m+1}}
        \begin{vmatrix}
            1 & 1 & \cdots & 1   \\
            - 2 & Lx+(2m-1) & & - 2   \\
            \vdots &  & \ddots & \vdots \\
            -2 & -2  & \cdots & Lx+(2m-1)
        \end{vmatrix}
        \\
        = \; & \frac{(x-1)}{L^{2m}}
        \begin{vmatrix}
            1 & 1 & \cdots & 1   \\
            0 & Lx+(2m+1) & & 0   \\
            \vdots &  & \ddots & \vdots \\
            0 & 0  & \cdots & Lx+(2m+1)
        \end{vmatrix}
        \\
        = \; &(x-1)(x+1)^{L-1}.
    \end{align}
    Next, we examine the case of $ k=1,\dots,L-1 $. Let \( \zeta = \zeta_N \). Initially, we verify that $1$ is a root of $f^{M,L}_{L,k}(x)$.
    \begin{align}
    &
    \begin{vmatrix}
        L + (2m-1) \z^{mk} & - 2 \z^{mk} & \cdots & - 2 \z^{mk} & \cdots & - 2 \z^{mk}   \\
        - 2 \z^{(m-1)k} & L + (2m-1) \z^{(m-1)k} & & - 2 \z^{(m-1)k}  & \cdots & - 2 \z^{(m-1)k}   \\
        \vdots &  & \ddots & \vdots \\
        -2 & -2 & \cdots & L + (2m-1) & \cdots & -2  \\
        \vdots & && \vdots & \ddots \\
        -2 \z^{-mk} & -2 \z^{-mk} & \cdots & -2 \z^{-mk}  & \cdots & L + (2m-1) \z^{-mk}
    \end{vmatrix}
    \\
    = \; &
    \begin{vmatrix}
        L \z^{-mk} + (2m-1) & - 2 & \cdots & - 2 & \cdots & - 2   \\
        - 2 & L \z^{-(m-1)k} + (2m-1) & & - 2  & \cdots & - 2  \\
        \vdots &  & \ddots & \vdots \\
        -2 & -2 & \cdots & 4m & \cdots & -2  \\
        \vdots & && \vdots & \ddots \\
        -2 & -2 & \cdots & -2  & \cdots & L \z^{mk} + (2m-1)
    \end{vmatrix}
    \\
    = \; &
    \begin{vmatrix}
        L \z^{-mk} + L & 0 & \cdots & -2L & \cdots & 0   \\
        0 & L \z^{-(m-1)k} + L &  & -2L & \cdots & 0  \\
        \vdots &  & \ddots & \vdots \\
        -2 & -2 & \cdots & 4m & \cdots & -2  \\
        \vdots & && \vdots & \ddots \\
        0 & 0 & \cdots & -2L  & \cdots & L \z^{mk} + L
    \end{vmatrix}
    \\
    = \; & L^{2m}
    \begin{vmatrix}
        \z^{-mk} + 1 & 0 & \cdots & -2 & \cdots & 0   \\
        0 &  \z^{-(m-1)k} + 1 &  & -2 & \cdots & 0  \\
        \vdots &  & \ddots & \vdots \\
        -2 & -2 & \cdots & 4m & \cdots & -2  \\
        \vdots & && \vdots & \ddots \\
        0 & 0 & \cdots & -2  & \cdots & \z^{mk} + 1
    \end{vmatrix}
    \\
    = \; & L^{2m}
    \begin{vmatrix}
        \z^{-mk} + 1 & 0 & \cdots & \z^{-mk} -1 & \cdots & 0   \\
        0 &  \z^{-(m-1)k} + 1 &  & \z^{-(m-1)k} -1 & \cdots & 0  \\
        \vdots &  & \ddots & \vdots \\
        -2 & -2 & \cdots & 0 & \cdots & -2  \\
        \vdots & && \vdots & \ddots \\
        0 & 0 & \cdots & \z^{mk} -1  & \cdots & \z^{mk} + 1
    \end{vmatrix}.
\end{align}
By using cofactor expansion, the minor determinant corresponds to the \(j\)\,th left (assume $ 0 \le j \le m-1 $) $-2$ is
\[
    (-1)^{m-j} ( \z^{-(m-j)k}-1 ) \prod_{ \substack{ 0 \le |l| \le m-1 \\ l \neq j } } ( \z^{-(m-l)k}+1 ).
\]
Similarly,  the minor determinant corresponds to the \(j\)\,ths right (assume $ 1 \le j \le m $) $-2$ is represented by the following formula.
\[
    (-1)^{m-j} ( \z^{(m-j)k}-1 ) \prod_{ \substack{ 0 \le |l| \le m-1 \\ l \neq j } } ( \z^{(m-l)k}+1 ).
\]
Here, we know
\begin{align}
    &(-1)^{m-j} ( \zeta^{-(m-j)k}-1 ) \prod_{ \substack{ 0 \le |l| \le m-1 \\ l \neq j } } ( \zeta^{-(m-l)k}+1 ) \\
    = \;&(-1)^{m-j} ( \zeta^{-(m-j)k}-1 ) ( \zeta^{(m-j)k}+1 ) \prod_{ \substack{ 0 \le |l| \le m-1 \\ l \neq \pm j } } ( \zeta^{-(m-l)k}+1 )  \\
    = \;&(-1)^{m-j} ( \zeta^{-(m-j)k} - \zeta^{(m-j)k} ) \prod_{ \substack{ 0 \le |l| \le m-1 \\ l \neq \pm j } } ( \zeta^{-(m-l)k}+1 )
\end{align}
and
\begin{align}
    &(-1)^{m-j} ( \zeta^{(m-j)k}-1 ) \prod_{ \substack{ 0 \le |l| \le m-1 \\ l \neq j } } ( \zeta^{(m-l)k}+1 ) \\
    = \; & (-1)^{m-j} ( \zeta^{(m-j)k}-1 ) ( \zeta^{-(m-j)k}+1 ) \prod_{ \substack{ 0 \le |l| \le m-1 \\ l \neq \pm j } } ( \zeta^{(m-l)k}+1 )  \\
    = \; & (-1)^{m-j} ( \zeta^{(m-j)k} - \zeta^{-(m-j)k} ) \prod_{ \substack{ 0 \le |l| \le m-1 \\ l \neq \pm j } } ( \zeta^{(m-l)k}+1 )    \\
    = & - (-1)^{m-j} ( \zeta^{-(m-j)k} - \zeta^{(m-j)k} ) \prod_{ \substack{ 0 \le |l| \le m-1 \\ l \neq \pm j } } ( \zeta^{-(m-l)k}+1 ).
\end{align}
Thus, the sum of these two terms is $0$ and the determinant is equal to $0$. Therefore, we know that $1$ is a root of $ f^{M,L}_{L,k}(x) $.

Also, for the case of $ \z^{lk} $ $(l=1,\dots,L-1)$, the following equality holds:
\begin{align}
    &
    \begin{vmatrix}
        L \zeta^{lk} + (2m-1) \zeta^{mk} & \cdots & - 2 \zeta^{mk} & \cdots & - 2 \zeta^{mk}   \\
        \vdots & \ddots & \vdots & & \vdots \\
        -2 & \cdots & L \zeta^{lk} + (2m-1) & \cdots & -2  \\
        \vdots & & \vdots & \ddots & \vdots \\
        -2 \zeta^{-mk} & \cdots & -2 \zeta^{-mk}  & \cdots & L \zeta^{lk} + (2m-1) \zeta^{-mk}
    \end{vmatrix}
    \\
    = \; &
        \begin{vmatrix}
        L \zeta^{(l-m)k} + (2m-1) & \cdots & - 2 & \cdots & - 2  \\
        \vdots & \ddots & \vdots & & \vdots \\
        -2 & \cdots & L \zeta^{lk} + (2m-1) & \cdots & -2  \\
        \vdots & & \vdots & \ddots & \vdots \\
        -2 & \cdots & -2  & \cdots & L \zeta^{(l+m)k} + (2m-1)
    \end{vmatrix}
\end{align}
and this is equal to the determinant for the case of $ x=1 $ since \( ( \z^{(l-m)k} + (2m-1), \z^{(l-(m-1))k} + (2m-1), \dots, \z^{(l+m)k} + (2m-1) ) \) is a rotation of $ ( \z^{-mk} + (2m-1), \z^{-(m-1)k} + (2m-1), \dots, \z^{mk} + (2m-1) ) $. Hence the all of roots of $ f^{M,L}_{L,k}(x) $ is $ \{ \z^{lk} \in \mathbf{C} \mid 0 \le l \le L-1 \} $ which is equal to $ \{ \z^{l} \in \mathbf{C} \mid 0 \le l \le L-1 \} $. This means \( f^{M,L}_{L,k}(x) = x^L-1 \) for \( k = 1,\dots,L-1 \).
\end{proof}
\begin{proof}[Proof of Lemma \ref{Mdevide}]
    As noted above, it holds that
    \[
        f^{M,L}_{N_1}(x) = \prod_{k=0}^{N_1-1} f^{M,L}_{N_1,k}(x)
    \]
    and
    \[
        f^{M,L}_{N_2}(x) = \prod_{k=0}^{N_2-1} f^{M,L}_{N_2,k}(x).
    \]
    Now, we see that
    \begin{align}
        f^{M,L}_{N_2}(x) = \; & \prod_{k=0}^{N_2-1} f^{M,L}_{N_2,k}(x)   \\
        = \; & \prod_{ \substack{ N_2 / N_1 \mid k \\ 0 \le k < N_2 } } f^{M,L}_{N_2,k}(x) \prod_{ \substack{ N_2 / N_1 \nmid k \\ 0 \le k < N_2 } } f^{M,L}_{N_2,k}(x) \\
        = \; & \prod_{ k=0 }^{N_1-1} f^{M,L}_{N_1,k}(x) \prod_{ \substack{ N_2 / N_1 \nmid k \\ 0 \le k < N_2 } } f^{M,L}_{N_2,k}(x)    \\
        = \; & f^{M,L}_{N_1}(x) \prod_{ \substack{ N_2 / N_1 \nmid k \\ 0 \le k < N_2 } } f^{M,L}_{N_2,k}(x).
    \end{align}
\end{proof}

\begin{proof}[Proof of Lemma \ref{Mfactorization}]
    In this proof, we write just $ \z $ instead of $ \z_N $. To begin with, the constant term of $ f^{M,L}_{N,k}(x) $ can be calculated as follows:
\begin{align}
    &\det \frac{1}{L}
    \begin{bmatrix}
        (2m-1) \z^{mk} & - 2 \z^{mk} & \cdots & - 2 \z^{mk}   \\
        - 2 \z^{(m-1)k} & (2m-1) \z^{(m-1)k} & & - 2 \z^{(m-1)k}   \\
        \vdots &  & \ddots & \vdots \\
        -2 \z^{-mk} & -2 \z^{-mk} & \cdots & (2m-1) \z^{mk}
    \end{bmatrix}
    \\
    = \; & \frac{1}{L^{2m+1}}
    \begin{vmatrix}
        (2m-1) \z^{mk} & - 2 \z^{mk} & \cdots & - 2 \z^{mk}   \\
        - 2 \z^{(m-1)k} & (2m-1) \z^{(m-1)k} & & - 2 \z^{(m-1)k}   \\
        \vdots &  & \ddots & \vdots \\
        -2 \z^{-mk} & -2 \z^{-mk} & \cdots & (2m-1) \z^{mk}
    \end{vmatrix}
    \\
    = \; & \frac{1}{L^{2m+1}}
    \begin{vmatrix}
        2m-1 & - 2 & \cdots & - 2   \\
        - 2 & 2m-1 & & - 2   \\
        \vdots &  & \ddots & \vdots \\
        -2 & -2  & \cdots & 2m-1
    \end{vmatrix}
    \\
    = \; & \frac{1}{L^{2m+1}}
    \begin{vmatrix}
        -(2m+1) & -(2m+1) & \cdots & - (2m+1)   \\
        - 2 & 2m-1 & & - 2   \\
        \vdots &  & \ddots & \vdots \\
        -2 & -2  & \cdots & 2m-1
    \end{vmatrix}
    \\
    = &-\frac{1}{L^{2m}}
    \begin{vmatrix}
        1 & 1 & \cdots & 1   \\
        - 2 & 2m-1 & & - 2   \\
        \vdots &  & \ddots & \vdots \\
        -2 & -2  & \cdots & 2m-1
    \end{vmatrix}
    \\
    = &-\frac{1}{L^{2m}}
    \begin{vmatrix}
        1 & 1 & \cdots & 1   \\
        0 & 2m+1 & & 0   \\
        \vdots &  & \ddots & \vdots \\
        0 & 0  & \cdots & 2m+1
    \end{vmatrix}
    \\
    = &-1.
\end{align}
Then, we calculate the coefficient of \(x\). For \( L x \) in the first row in \eqref{Mproduct}, the coefficient is the following determinant:
\begin{align}
    &\frac{1}{L^{2m+1}}
    \begin{vmatrix}
        (2m-1) \z^{(m-1)k} & - 2 \z^{(m-1)k} & \cdots & - 2 \z^{(m-1)k}   \\
        - 2 \z^{(m-2)k} & (2m-1) \z^{(m-2)k} & & - 2 \z^{(m-2)k}   \\
        \vdots &  & \ddots & \vdots \\
        -2 \z^{-mk} & -2 \z^{-mk} & \cdots & (2m-1) \z^{-mk}
    \end{vmatrix}
    \\
    = \; &\frac{\z^{-mk}}{L^{2m+1}}
    \begin{vmatrix}
        2m-1&-2&\cdots&-2\\
        -2&2m-1&&-2\\
        \vdots&&\ddots&\vdots\\
        -2&-2&\cdots&2m-1
    \end{vmatrix}
    \\
    = \; &\frac{\z^{-mk}}{L^{2m+1}}
    \begin{vmatrix}
        -(2m-1)&-(2m-1)&\cdots&-(2m-1)\\
        -2&2m-1&&-2\\
        \vdots&&\ddots&\vdots\\
        -2&-2&\cdots&2m-1
    \end{vmatrix}
    \\
    = &-\frac{2m-1}{L^{2m+1}}\z^{-mk}
    \begin{vmatrix}
        1&1&\cdots&1\\
        -2&2m-1&&-2\\
        \vdots&&\ddots&\vdots\\
        -2&-2&\cdots&2m-1
    \end{vmatrix}
    \\
    = &-\frac{2m-1}{L^{2m+1}}\z^{-mk}
    \begin{vmatrix}
        1&1&\cdots&1\\
        0&2m+1&&0\\
        \vdots&&\ddots&\vdots\\
        0&0&\cdots&2m+1
    \end{vmatrix}
    \\
    = &-\frac{2m-1}{L^2} \z^{-mk}.
\end{align}
Thus we get a term \( - [ (2m-1)/L^2 ] \z^{-mk} (Lx) \), i.e. \( - [ (2m-1)/L ] \z^{-mk} x \). We can apply this calculation for \( Lx \) in each row, and we obtain the first order term
\[
    - \frac{2m-1}{L} \bigg( \sum_{ |j| \le m } \z^{jk} \bigg) x.
\]
Next, we examine the coefficient of \(x^2\). There are \( 2m+1 \) \(Lx\)'s in the matrix. So there is \( \binom{2m+1}{2} \) \(x^2\) terms in our determinant. For the term whose \(Lx\)'s comes from the top two rows, its coefficient is given by the following:
\begin{align}
    &\frac{1}{L^{2m+1}}
    \begin{vmatrix}
        (2m-1) \z^{(m-2)k} & - 2 \z^{(m-2)k} & \cdots & - 2 \z^{(m-2)k}   \\
        - 2 \z^{(m-3)k} & (2m-1) \z^{(m-3)k} & & - 2 \z^{(m-3)k}   \\
        \vdots &  & \ddots & \vdots \\
        -2 \z^{-mk} & -2 \z^{-mk} & \cdots & (2m-1) \z^{-mk}
    \end{vmatrix}
    \\
    = \; &\frac{\z^{-(2m-1)k}}{L^{2m+1}}
    \begin{vmatrix}
        2m-1&-2&\cdots&-2\\
        -2&2m-1&&-2\\
        \vdots&&\ddots&\vdots\\
        -2&-2&\cdots&2m-1
    \end{vmatrix}
    \\
    = \; &\frac{\z^{-(2m-1)k}}{L^{2m+1}}
    \begin{vmatrix}
        -(2m-3)&-(2m-3)&\cdots&-(2m-3)\\
        -2&2m-1&&-2\\
        \vdots&&\ddots&\vdots\\
        -2&-2&\cdots&2m-1
    \end{vmatrix}
    \\
    = &-\frac{2m-3}{L^{2m+1}}\z^{-(2m-1)k}
    \begin{vmatrix}
        1&1&\cdots&1\\
        -2&2m-1&&-2\\
        \vdots&&\ddots&\vdots\\
        -2&-2&\cdots&2m-1
    \end{vmatrix}
    \\
    = &-\frac{2m-3}{L^{2m+1}}\z^{-(2m-1)k}
    \begin{vmatrix}
        1&1&\cdots&1\\
        0&2m+1&&0\\
        \vdots&&\ddots&\vdots\\
        0&0&\cdots&2m+1
    \end{vmatrix}
    \\
    = &-\frac{2m-3}{L^3} \z^{-(2m-1)k}.
\end{align}
Thus we get a term \( - [ (2m-3) / L ] \z^{-(2m-1)k} x^2 \). Similarly, we can calculate the term for each combination and we find \( - [ (2m-3) / L ] \z^{jk} x^2 \) for some \(j\) which satisfies \( |j| \le 2m-1 \). Now, let \( a_j \) be the number of appearance of \( - [ (2m-3) / L ] \z^{jk} x^2 \). Of course, \(a_j\) is an ingeger for all \( j \). Then we obtain the term of degree \(2\):
\[
    -\frac{2m-3}{L} \biggl( \sum_{ |j| \le 2m-1 } a_j \z^{jk} \biggr) x^2.
\]
(In fact, \( a_j = m - q_2(|j|) \), where \( q_2(|j|) \) is the quotient of \(|j|\) divided by \(2\).)


Similarly, we know that the term of degree 3 is
\[
    -\frac{2m-5}{L} \biggl( \sum_{|j|\le3m-3} b_j \z^{jk} \biggr) x^3
\]
for some integers \( b_j, |j| \le 3m-3 \).
\end{proof}

\begin{proof}[Proof of Lemma \ref{Mcoprime}]
    Recall that
    \[
        f^{M,L}_{n}(x) = \prod_{k=0}^{n-1} f^{M,L}_{n,k}(x).
    \]
    In this case, the coefficient of $ x $ is not an integer. Now, by Lemma \ref{Mfactorization}, the coefficient can be written as follows:
    \begin{align*}
        (-1)^n \frac{2m-1}{L} \sum_{k=0}^{n-1} \sum_{ |j| \le m } \z_n^{jk} = \; & (-1)^n \frac{2m-1}{L} \sum_{ |j| \le m } \sum_{k=0}^{n-1} \z_n^{jk} \\
        = \; & (-1)^n \frac{2m-1}{L} \sum_{ \substack{ |j| \le m \\ n \mid j } } n  \\
        = \; & (-1)^n \frac{2m-1}{L} ( 2 q_n(m) + 1 ) n,
    \end{align*}
    where $ q_n(m) $ is the quotient of $m$ divided by $n$. Here, $ 2m-1 $ is coprime to $ L=2m+1 $ and $n$ is also coprime with it. Moreover, if $ n>m $, $ q_n(m)=0 $ and the coefficient is not an integer. Otherwise, $ 2 q_n(m) + 1 < 2m+1 = L $, so  $ L \nmid ( 2 q_n(m) + 1 ) $ and the coefficient is not an integer.
\end{proof}
We also give the proof of Lemma \ref{M_L^2} here.
\begin{proof}
    In this proof, $ \z $ represents \( e^{ 2 \pi i / L^2 } \).
In this case, the coefficient of \(x^3\) is not an integer. That is made of three types of products: i) a product of one \(x^3\) and \(L^2-1\) constants; ii) a product of one \( x^2 \), one \(x\), and \(L^2-2\) constants; iii) a product of three \(x\)'s and \(L^2-3\) constants. It is easy to see that the coefficients of the sum of all terms of type i) and ii) is an integer. Note that
\[
    \sum_{k=0}^{L^2-1} \z^{jk} = 0
\]
unless \( j \) is a multiple of \( N \). Then, for the sum of type i),
\begin{align}
    - \frac{2m-5}{L} \sum_{k=0}^{L^2-1} \sum_{|j|\le3m-3} b_j \z^{jk} = &- \frac{2m-5}{L} \sum_{|j|\le3m-3} b_j \sum_{k=0}^{L^2-1} \z^{jk} \\
    = &- \frac{2m-5}{L} \sum_{ \substack{ |j|\le3m-3 \\ L^2 \mid j } } b_j \sum_{k=0}^{L^2-1} \z^{jk} \\
    = &- \frac{2m-5}{L} b_0 L^2 \\
    = &- b_0 L (2m-5)
\end{align}
holds, and this is an integer.

Similarly, for the sum of type ii), we have
\begin{align}
    - \frac{(2m-1)(2m-3)}{L^2} &\sum_{ \substack{ 0 \le k,l < L^2 \\ k \neq l } } \biggl( \sum_{|j_1|\le2m-1} a_{j_1}  \z^{{j_1}k} \biggr) \biggl( \sum_{ |j_2| \le m } \z^{{j_2}l} \biggr).
\end{align}
Here, nonzero terms have the form
\[
    a_j \z^{j(k-l)}, \quad |j| \le m,
\]
because if \( j_1 \neq -j_2 \) and $ j_2 \neq 0 $,
\begin{align}
    \sum_{ k \neq l } \z^{ j_1 k } \z^{ j_2 l } = \; & \sum_{k=0}^{L^2-1} \z^{ j_1 k } \sum_{ l \neq k } \z^{ j_2 l }   \\
    = \; & \sum_{k=0}^{L^2-1} \z^{ j_1 k } ( - \z^{ j_2 k } )   \\
    = &- \sum_{k=0}^{L^2-1} \z^{(j_1+j_2)k}
\end{align}
and \( 0 < | j_1 + j_2 | \le 3m-1 < L^2 \). It is the same when $ j_1 \neq 0 $. Now, let \( j := j_1 \) and \( j_2 = -j \) and it holds that
\begin{align}
    \sum_{ k \neq l } \z^{ j (k-l) } = \; & \sum_{k=0}^{L^2-1} \z^{ j k } \sum_{ l \neq k } \z^{ - j l }   \\
    = \; & \sum_{k=0}^{L^2-1} \z^{ j k } ( - \z^{ - j k } )   \\
    = &- \sum_{k=0}^{L^2-1} 1   \\
    = &- L^2.
\end{align}
Therefore, the sum of type ii) is
\[
    ( 2m-1 ) ( 2m-3 ) \sum_{ |j| \le m } a_j
\]
and this is an integer. The problem is type iii):
\[
    -\frac{(2m-1)^3}{L^3} \sum_{ \substack{ 0 \le k,l,s < L^2 \\ k<l<s } } \biggl( \sum_{ |j_1| \le m } \z^{j_1k} \biggr) \biggl( \sum_{ |j_2| \le m } \z^{j_2l} \biggr) \biggl( \sum_{ |j_3| \le m } \z^{j_3s} \biggr).
\]
By symmetry, this is equal to the following:
\[
    - \frac{1}{6} \frac{(2m-1)^3}{L^3} \sum_{\substack{ \mathrm{dist}(k,l,s) }} \biggl( \sum_{ |j_1| \le m } \z^{j_1k} \biggr) \biggl( \sum_{ |j_2| \le m } \z^{j_2l} \biggr) \biggl( \sum_{ |j_3| \le m } \z^{j_3s} \biggr),
\]
where \( \mathrm{dist}(k,l,s) \) means \( k \neq l \), \( l \neq s \), and \( s \neq k \). After expanding the product of right three summations and considering the summation over \( \mathrm{dist}(k,l,s) \), almost all terms vanishes, because
\[
    \sum_{k=0}^{L^2-1} \z^{jk} = 0
\]
unless \( j \) is a multiple of \( L^2 \). Then there are three types of nonzero terms: type a)
\[
    \sum_{\mathrm{dist}(k,l,s)} 1,
\]
which came from three \(1\)'s; b)
\[
    \sum_{\mathrm{dist}(k,l,s)} \z^{jx-jy}
\]
for \( x \neq y, \{x,y\} \subset \{ k,l,s \} \), and \( j = 1, \dots, m \); c)
\[
    \sum_{\mathrm{dist}(k,l,s)} \z^{jx-\tilde{j}y-( j - \tilde{j} )z}
\]
for \( \{ x,y,z \} = \{ k,l,s \} \), \( j=2,\dots,m \), and \( \tilde{j}=1,\dots,j-1 \).

For type a),
\begin{align}
    \sum_{ \substack{ 0 \le k,l,s < L^2 \\ \mathrm{dist}(k,l,s) } } 1 = \; & L^2 ( L^2-1 )( L^2-2 )
\end{align}
holds. Consequently, the sum of type a) is \( L^2 ( L^2-1 )( L^2-2 ) \).

For type b),
\begin{align}
    \sum_{\mathrm{dist}(k,l,s)} \z^{jx-jy} = \; & \sum_{x=0}^{L^2-1} \z^{jx} \sum_{ y \neq x } \z^{-jy} \cdot ( L^2-2 )   \\
    = \; & ( L^2 - 2 ) \sum_{x=0}^{L^2-1} \z^{jx} ( - \z^{-jx} )    \\
    = &- ( L^2 - 2 ) \sum_{x=0}^{L^2-1} 1 \\
    = &- L^2 ( L^2 - 2 )
\end{align}
holds. Moreover, there are \( 6 m \) type b) terms (\((x,y)=(k,l),(l,k),...\) and \( j=1,\dots,m \)). Consequently, the sum of type b) is \( - 6 m L^2 ( L^2 - 2 ) \).

For type c),
\begin{align}
    \sum_{\mathrm{dist}(k,l,s)} \z^{jx-\tilde{j}y-( j - \tilde{j} )z} = \; & \sum_{x=0}^{L^2-1} \z^{jx} \sum_{ y \neq x } \z^{ - \tilde{j} y } \sum_{ \substack{ z \neq x \\ z \neq y } } \z^{ -( j - \tilde{j} )z }    \\
    = \; & \sum_{x=0}^{L^2-1} \z^{jx} \sum_{ y \neq x } \z^{ - \tilde{j} y } ( - \z^{ -( j - \tilde{j} )x } - \z^{ -( j - \tilde{j} )y } )  \\
    = &- \sum_{x=0}^{L^2-1} \z^{ \tilde{j} x} \sum_{ y \neq x } \z^{ - \tilde{j} y } - \sum_{x=0}^{L^2-1} \z^{ {j} x} \sum_{ y \neq x } \z^{ - j y }    \\
    = &- \sum_{x=0}^{L^2-1} \z^{ \tilde{j} x} ( - \z^{ - \tilde{j} x } ) - \sum_{x=0}^{L^2-1} \z^{ {j} x} ( - \z^{ - {j} x } )  \\
    = \; & 2 L^2
\end{align}
holds. Moreover, there are \( 6 m (m-1) / 2 \) type c) terms. Consequently, the sum of type c) is \( 6 m (m-1) L^2 \).

Thus the sum of type iii) can be shown as follows:
\begin{align}
    - \frac{1}{6} \frac{ (2m-3)^3 }{L^3} ( L^2 ( L^2-1 ) ( L^2 - 2 ) -6m ( L^2-2 ) L^2 + 6m(m-1)L^2 ).
\end{align}
This is equal to
\[
    - \frac{1}{6} \biggl( (2m-1)^3 ( 8m^3-3m-1 ) + \frac{ (2m-1)^3 (m+1) }{ 2m+1 } \biggr).
\]
This is not an integer because $ 2m-1 $ is coprime to $2m+1$ and $ m+1 < 2m+1 $.

Therefore, the sum of type iii) is not an integer, while that of type i) and ii) are integers. Consequently, the coefficient of degree \(3\) of \( f^{M,L}_{L^2}(x) \) is not an integer. This completes the proof.
\end{proof}


\subsection{F-type Grover walks}

Similar lemmas and theorem hold as in the previous section. The following notation is defined as in the case of  M type.

\begin{notation} For \( L,N \ge 2 \), \( k=0,\dots,L-1 \),
    \[
        f^{F,L}_{N,k}(x) := \det( x I_L - Z_L^k A^{F,L} ),
    \]
    where \(A^{F,L}\) is defined in subsection \ref{GroverWalks}, and recall that
    \[
        Z_L =
        \begin{bmatrix}
            \z_N^m & 0 & \cdots & 0 \\
            0 & \z_N^{m-1} & & 0 \\
            \vdots & & \ddots & \vdots \\
            0 & 0 & \cdots & \z_N^{-m}
        \end{bmatrix}
    \]
    and $ \z_N := e^{ 2 \pi i / N } $.  That is, $ f^{F,L}_{N,k}(x) $ can be written as
    \[
        \det \frac{1}{L}
        \begin{bmatrix}
        L x - 2 \z_N^{mk} & & \cdots & & -2 \z_N^{mk} & (2m-1) \z_N^{mk}   \\
        - 2 \z_N^{(m-1)k} & \ddots & & & (2m-1) \z_N^{(m-1)k} & - 2 \z_N^{(m-1)k}   \\
        \vdots &  & & & \vdots & \vdots \\
        -2  & & L x + (2m-1) & & -2 & - 2   \\
        \vdots & & & \ddots & \vdots & \vdots \\
        -2 \z_N^{-(m-1)k} & & \cdots & & L x - 2 \z_N^{-(m-1)k} & -2 \z_N^{-(m-1)k} \\
        (2m-1) \z_N^{-mk} & & \cdots &  & -2 \z_N^{-mk} & L x - 2 \z_N^{-mk}
    \end{bmatrix}.
    \]
\end{notation}

Note that
\[
    f^{F,L}_N(x) = \prod_{k=0}^{N-1} f^{F,L}_{N,k}(x).
\]
As in the case of M type, see \cite{AKST2024} for a proof. The following is our main theorems for F-type Grover walk.

\begin{thm}\label{Fmainthm}
    Let \( L \ge 3 \) be a prime number. Then,
    \[
        T^{F,L}_N =
        \begin{cases}
            4, & N=L,    \\
            \infty, & N \neq L
        \end{cases}
    \]
    for all \( N \ge 2 \).
\end{thm}
\begin{thm}\label{Fmainthm2}
    Let \( L \ge 3 \) be an odd number. Then $T^{F,L}_L = 4$ and $ T^{F,L}_N = \infty $ if $N$ has a factor which is coprime to $L$. 
\end{thm}
To prove these theorem, we prepare the following lemmas.
\begin{lem}\label{Fdevide}
    Let $ L, N_1, N_2 \ge 2 $. If $ N_1 \mid N_2 $, then $ f^{F,L}_{N_1}(x) \mid f^{F,L}_{N_2}(x) $ as polynomials.
\end{lem}
\begin{lem}\label{Ffactorization}
    Let \( L \ge 3 \) be an odd number. Then the coefficients of $ f^{F,L}_{N,k}(x) $ are partially clarified as follows:
    \begin{multline}
        f^{F,L}_{N,k}(x) = \ x^{L} + \dots +\frac{(-1)^{m}}{L} \biggl( - 2 (m-1) \sum_{ 0 < |j| \le m } \z^{jk} + m(2m-5) \biggr) x^3 \\
        + \frac{(-1)^{m}}{L} \biggl( 2 \sum_{ 0 < |j| < m } \z^{jk} + m(2m-3) \biggr) x^2 \\
        + \frac{(-1)^{m+1}}{L} \bigg( 2 \sum_{ 0 < |j| \le m } \z^{jk} + 2m-1 \bigg) x + (-1)^{m+1}.
    \end{multline}
\end{lem}
\begin{lem}\label{Fcoprime}
    Let  $ L $ be an odd number. If $ n \ge 2 $ is coprime to $ L $, then $ f^{F,L}_{n}(x) $ is not an element of \( \mathbf{Z}[x] \).
\end{lem}
\begin{lem}\label{F_L^2}
    Let \( L \) be an odd prime. Then $ f^{F,L}_{L^2}(x) $ is not an element of \( \mathbf{Z}[x] \).
\end{lem}
%
The case of $ N = L $ of Theorem \ref{Fmainthm} is a special case of the corresponding part of Theorem \ref{Fmainthm2}. The case of $ N \neq L $ of Theorem \ref{Fmainthm} can be proved in the same way as in the case of the M type by using the above lemmas, and the same applies to Lemma \ref{Fdevide}. Therefore, the proofs are omitted. We now proceed to the proof of Theorem \ref{Fmainthm2}.

\begin{proof}[Proof of Theorem \ref{Fmainthm2}]
    The latter part can be proved as in the case of Theorem \ref{Fmainthm}. To prove the former part, it is enough to show that
    \[
        ( Z_L^k A^{F,L} )^4 = I_L
    \]
    for all \( k = 0, \dots, L-1 \) and
    \[
        ( Z_L^k A^{F,L} )^2 \neq I_L
    \]
    for at least one \( k \). It is easy to see that
    \[
        ( Z_L^0 A^{F,L} )^2 = I_L
    \]
    and of course,
    \[
        ( Z_L^0 A^{F,L} )^4 = I_L.
    \]
    Let \( k \neq 0 \) and \( a^{(k)}_{ij} \) represent the \((i,j)\) component of the matrix \( ( Z_L^k A^{F,L} )^2 \). Then we know
    \[
        a^{(k)}_{ij} =
        \begin{cases}
            L^2-4L,     & i=j,    \\
            -2L ( 1 + \z^{(j-i)k} ), & i \neq j
        \end{cases}
    \]
    by direct calculation. Obviously, $ ( Z_L^k A^{F,L} )^2 \neq I_L $. Furthermore, it can be directly checked again that
    \[
        ( Z_L^k A^{F,L} )^4 = ( ( Z_L^k A^{F,L} )^2 )^2 = I_L.
    \]
\end{proof}

Then we will give proofs of Lemmas \ref{Ffactorization}, \ref{Fcoprime}, and \ref{F_L^2}.

\begin{proof}[Proof of Lemma \ref{Ffactorization}]
Recall that the definition of $ f^{F,L}_{N,k} $ is as follows:
\begin{equation}\label{Ffactor}
    f^{F,L}_{N,k}(x) :=
    \det \frac{1}{L}
    \begin{bmatrix}
    L x - 2 \z_N^{mk} & & -2 \z_N^{mk} & (2m-1) \z_N^{mk}   \\
    - 2 \z_N^{(m-1)k} & \ddots & (2m-1) \z_N^{(m-1)k} & - 2 \z_N^{(m-1)k}   \\
    \vdots &  & \ddots & \vdots \\
    (2m-1) \z_N^{-mk} &  & -2 \z_N^{-mk} & L x - 2 \z_N^{-mk}
\end{bmatrix}.
\end{equation}
To begin with, the constant term can be calculated as follows:
\begin{align}
    &\det \frac{1}{L}
    \begin{bmatrix}
        - 2 \z^{mk} & & -2 \z^{mk} & (2m-1) \z^{mk}   \\
        - 2 \z^{(m-1)k} & \ddots & (2m-1) \z^{(m-1)k} & - 2 \z^{(m-1)k}   \\
        \vdots &  & \ddots & \vdots \\
        (2m-1) \z^{-mk} &  & -2 \z^{-mk} & - 2 \z^{-mk}
    \end{bmatrix}
    \\
    = \; & \frac{1}{L^{2m+1}}
    \begin{vmatrix}
        - 2 \z^{mk} & & -2 \z^{mk} & (2m-1) \z^{mk}   \\
        - 2 \z^{(m-1)k} & \ddots & (2m-1) \z^{(m-1)k} & - 2 \z^{(m-1)k}   \\
        \vdots &  & \ddots & \vdots \\
        (2m-1) \z^{-mk} &  & -2 \z^{-mk} & - 2 \z^{-mk}
    \end{vmatrix}
    \\
    = \; & \frac{1}{L^{2m+1}}
    \begin{vmatrix}
        - 2 & & -2 & 2m-1   \\
        - 2 & \ddots & 2m-1 & - 2   \\
        \vdots &  & \ddots & \vdots \\
        2m-1 &  & -2 & - 2
    \end{vmatrix}
    \\
    = \; & \frac{1}{L^{2m+1}}
    \begin{vmatrix}
        - (2m+1) & & -(2m+1) & -(2m+1)   \\
        - 2 & \ddots & 2m-1 & - 2   \\
        \vdots &  & \ddots & \vdots \\
        2m-1 &  & -2 & - 2
    \end{vmatrix}
    \\
    = &-\frac{1}{L^{2m}}
    \begin{vmatrix}
        1 & & 1 & 1   \\
        - 2 & \ddots & 2m-1 & - 2   \\
        \vdots &  & \ddots & \vdots \\
        2m-1 &  & -2 & - 2
    \end{vmatrix}
    \\
    = &-\frac{1}{L^{2m}}
    \begin{vmatrix}
        1 & & 1 & 1   \\
        0 & \ddots & 2m+1 & 0   \\
        \vdots &  & \ddots & \vdots \\
        2m+1 &  & 0 & 0
    \end{vmatrix}
    \\
    = \; &(-1)^{m+1}.
\end{align}
Then, we calculate the coefficient of \(x\). For \( L x \) in the first row in \eqref{Ffactor}, the coefficient is the following determinant:
\begin{align}
    &\frac{1}{L^{2m+1}}
    \begin{vmatrix}
        - 2 \z^{(m-1)k} & \cdots & - 2 \z^{(m-1)k} & (2m-1) \z^{(m-1)k} & - 2 \z^{(m-1)k}   \\
        - 2 \z^{(m-2)k} & \cdots & (2m-1) \z^{(m-2)k} & - 2 \z^{(m-2)k} & - 2 \z^{(m-2)k} \\
        \vdots &  &  & \vdots & \vdots \\
        (2m-1) \z^{-(m-1)k} &  & -2 \z^{-(m-1)k} & - 2 \z^{-(m-1)k} & -2 \z^{-(m-1)k}   \\
        -2 \z^{-mk} & \cdots & -2 \z^{-mk} & - 2 \z^{-mk} & -2 \z^{-mk}
    \end{vmatrix}
    \\
    = \; &\frac{\z^{-mk}}{L^{2m+1}}
    \begin{vmatrix}
        - 2 & \cdots & - 2 & 2m-1 & - 2   \\
        - 2 & \cdots & 2m-1 & - 2 & - 2 \\
        \vdots &  &  & \vdots & \vdots \\
        2m-1 &  & -2 & - 2 & -2 \\
        -2 & \cdots & -2 & -2 & -2
    \end{vmatrix}
    \\
    = \; &\frac{\z^{-mk}}{L^{2m+1}}
    \begin{vmatrix}
        0 & \cdots & 0 & 2m+1 & 0   \\
        0 & \cdots & 2m+1 & 0 & 0 \\
        \vdots &  &  & \vdots & \vdots \\
        2m+1 &  & 0 & 0 & 0 \\
        -2 & \cdots & -2 & -2 & -2
    \end{vmatrix}
    \\
    = \; &\frac{-2}{L^{2m+1}}\z^{-mk}
    \begin{vmatrix}
        0 & \cdots & 0 & 2m+1   \\
        0 & \cdots & 2m+1 & 0 \\
        \vdots &  &  & \vdots \\
        2m+1 &  & 0 & 0
    \end{vmatrix}
    \\
    = \; &(-1)^{m+1}\frac{2}{L^2} \z^{-mk}.
\end{align}
Thus we get a term \( (-1)^{m+1} ( 2/L ) \z^{-mk} x \) for the \( x \) in the first row in \eqref{Ffactor}. Similarly, we obtained \( (-1)^{m+1} ( 2/L ) \z^{-jk} x \) for \( 0 < |j| \le m \). Moreover, for \( Lx \) in the \(m\)\,th row in \eqref{Ffactor}, the coefficient is the following determinant:
\begin{align}
    &\frac{1}{L^{2m+1}}
    \begin{vmatrix}
        - 2 \z^{mk} & \cdots & - 2 \z^{mk} & (2m-1) \z^{mk}   \\
        - 2 \z^{(m-1)k} & \cdots & (2m-1) \z^{(m-1)k} & - 2 \z^{(m-1)k}  \\
        \vdots &  &  & \vdots  \\
        (2m-1) \z^{-mk} & \cdots & -2 \z^{-mk} & - 2 \z^{-mk} 
    \end{vmatrix}
    \\
    = \; &\frac{1}{L^{2m+1}}
    \begin{vmatrix}
        - 2 & \cdots & - 2 & 2m-1    \\
        - 2 & \cdots & 2m-1 & - 2  \\
        \vdots &  &  & \vdots  \\
        2m-1 & \cdots & -2 & - 2 
    \end{vmatrix}
    \\
    = \; &\frac{1}{L^{2m+1}}
    \begin{vmatrix}
        - (2m-1) & \cdots & - (2m-1) & -(2m-1)    \\
        - 2 & \cdots & 2m-1 & - 2  \\
        \vdots &  &  & \vdots  \\
        2m-1 & \cdots & -2 & - 2 
    \end{vmatrix}
    \\
    = &-\frac{2m-1}{L^{2m+1}}
    \begin{vmatrix}
        1 & \cdots & 1 & 1    \\
        - 2 & \cdots & 2m-1 & - 2  \\
        \vdots &  &  & \vdots  \\
        2m-1 & \cdots & -2 & - 2 
    \end{vmatrix}
    \\
    = &-\frac{2m-1}{L^{2m+1}}
    \begin{vmatrix}
        1 & \cdots & 1 & 1    \\
        0 & \cdots & 2m+1 & 0  \\
        \vdots &  &  & \vdots  \\
        2m+1 & \cdots & 0 & 0 
    \end{vmatrix}
    \\
    = \; &(-1)^{m+1}\frac{2m-1}{L^2}.
\end{align}
As a result, we obtain the first order term
\[
    \frac{(-1)^{m+1}}{L} \bigg( 2 \sum_{ 0 < |j| \le m } \z^{jk} + 2m-1 \bigg) x,
\]
or
\[
    \frac{(-1)^{m+1}}{L} \bigg( 2 \sum_{ |j| \le m } \z^{jk} + 2m-3 \bigg) x.
\]

Next, we examine the coefficient of \(x^2\). If we focus on \(Lx\) in a symmetric position, the coefficient is given by the following determinant.
\begin{align}
    &\frac{1}{L^{2m+1}}
    \begin{vmatrix}
        - 2 & \cdots & - 2 & 2m-1    \\
        - 2 & \cdots & 2m-1 & - 2  \\
        \vdots &  &  & \vdots  \\
        2m-1 & \cdots & -2 & - 2 
    \end{vmatrix}
    \\
    = \; &\frac{1}{L^{2m+1}}
    \begin{vmatrix}
        - (2m-3) & \cdots & - (2m-3) & -(2m-3)    \\
        - 2 & \cdots & 2m-1 & - 2  \\
        \vdots &  &  & \vdots  \\
        2m-1 & \cdots & -2 & - 2 
    \end{vmatrix}
    \\
    = &-\frac{2m-3}{L^{2m+1}}
    \begin{vmatrix}
        1 & \cdots & 1 & 1    \\
        - 2 & \cdots & 2m-1 & - 2  \\
        \vdots &  &  & \vdots  \\
        2m-1 & \cdots & -2 & - 2 
    \end{vmatrix}
    \\
    = &-\frac{2m-3}{L^{2m+1}}
    \begin{vmatrix}
        1 & \cdots & 1 & 1    \\
        0 & \cdots & 2m+1 & 0  \\
        \vdots &  &  & \vdots  \\
        2m+1 & \cdots & 0 & 0 
    \end{vmatrix}
    \\
    = \; &(-1)^{m}\frac{2m-3}{L^3}.
\end{align}
And there are $m$ symmetrical choices. Then, the case where we choose \(Lx\) in the \(m\)\,th row is similar to the former case of first order terms. For example, if we choose \(Lx\) in the \(m\)\,th row and the first row, we obtain the term \( (-1)^{m} (2/L) \z^{-mk} x^2 \). In other cases, the size of the matrix is \( (2m-1) \times (2m-1) \) and there are \( (2m-3) \) \((2m-1)\)'s, and the remaining elements are all 2. Thus, there are two rows that are exactly the same, and the determinant is equal to \(0\). Then we get the term of degree \(2\):
\[
    \frac{(-1)^{m}}{L} \biggl( 2 \sum_{ 0 < |j| \le m } \z^{jk} + m(2m-3) \biggr) x^2,
\]
or
\[
    \frac{(-1)^{m}}{L} \biggl( 2 \sum_{ |j| \le m } \z^{jk} + 2m^2-3m-2 \biggr) x^2.
\]

Finally, consider the case of the third degree. If we choose the \(m\)\,th \(Lx\) and choose the other two symmetrically, we obtain the term \( (-1)^{m-2} [ (2m-5) / L ] x^3 \). If we choose the other two not symmetrically, the coefficient is \(0\). Then if we do not choose the \(m\)\,th \(Lx\) and do not choose any symmeric position, there is too few \( 2m-1 \) and the determinant is equal to \(0\). If we choose a pair of two \(Lx\)'s, we have \( (-1)^{m-1} ( 2 / L ) \z^{jk}, \ 0 < |j| \le m \), and for each \(j\), the number of choice of pair is \( m-2 \). Then we get the term of degree \(3\):
\[
    \frac{(-1)^{m}}{L} \biggl( - 2 (m-1) \sum_{ 0 < |j| \le m } \z^{jk} + m (2m-5) \biggr) x^3,
\]
or
\[
    \frac{(-1)^{m}}{L} \biggl( - 2 (m-1) \sum_{ |j| \le m } \z^{jk} + 2m^2-3m-2 \biggr) x^3.
\]
\end{proof}

\begin{proof}[Proof of Lemma \ref{Fcoprime}]
    Recall that
    \[
        f^{F,L}_{n}(x) = \prod_{k=0}^{n-1} f^{F,L}_{n,k}(x).
    \]
    In this case, the coefficient of $ x $ is not an integer. Now, by Lemma \ref{Ffactorization}, the coefficient can be written as follows:
    \begin{align*}
        \frac{(-1)^{n(m+1)}}{L} \sum_{k=0}^{n-1} \bigg( 2 \sum_{ 0 < |j| \le m } \z^{jk} + 2m-1 \bigg) = \; & \frac{(-1)^{n(m+1)}}{L} \sum_{k=0}^{n-1} \bigg( 2 \sum_{ |j| \le m } \z^{jk} + 2m-3 \bigg)    \\
        = \; & \frac{(-1)^{n(m+1)}}{L} \bigg( 2 \sum_{k=0}^{n-1} \sum_{ |j| \le m } \z^{jk} + \sum_{k=0}^{n-1} (2m-3) \bigg) \\
        = \; & \frac{(-1)^{n(m+1)}}{L} \bigg( 2  \sum_{ \substack{ |j| \le m \\ n \mid j } } n + (2m-3)n \bigg) \\
        = \; & \frac{(-1)^{n(m+1)}n}{L} ( 2 ( 2 q_n(m) + 1) + 2m-3 ).
    \end{align*}
    Recall that $ q_n(m) $ is the quotient of $m$ divided by $n$. Here,
    \[
        m = n q_n(m) + r_n(m) \ge n q_n(m) \ge 2 q_n(m)
    \]
    holds, where $ r_n(m) $ represents the remainder of $m$ divided by $n$. Thus,
    \[
        2 ( 2 q_n(m) + 1) + 2m-3 \le 2( m+1 ) + 2m - 3 = 2L-3
    \]
    holds. Therefore, for this to be a multiple of $L$, it must be equal to $L$. The condition can be expressed as
    \[
        2 ( 2 q_n(m) + 1) + 2m-3 = 2m + 1
    \]
    and this is equivalent to
    \[
        2 q_n(m) = 1.
    \]
    However, this cannot hold since $ q_n(m) $ is an integer. Consequently, $ 2 ( 2 q_n(m) + 1) + 2m-3 $ is not a multiple of $L$. Moreover, $n$ is also coprime to $L$ by the assumption. Thus, the coefficient of $x$ in $ f^{F,L}_{n}(x) $,
    \[
        \frac{(-1)^{n(m+1)}n}{L} ( 2 ( 2 q_n(m) + 1) + 2m-3 ),
    \]
    is not an integer.    
\end{proof}

\begin{proof}[Proof of Lemma \ref{F_L^2}]
By Lemma \ref{Ffactorization}, the characteristic polynomial \( f^{F,L}_{L^2}(x) \) can be written as follows:
\begin{multline}
    f^{F,L}_{L^2}(x) = \prod_{k=0}^{L^2-1} \Biggl( x^{L} + \dots +\frac{(-1)^{m}}{L} \biggl( - 2 (m-1) \sum_{ 0 < |j| \le m } \zeta^{jk} + m(2m-5) \biggr) x^3 \\
        + \frac{(-1)^{m}}{L} \biggl( 2 \sum_{ 0 < |j| < m } \zeta^{jk} + m(2m-3) \biggr) x^2 \\
        + \frac{(-1)^{m+1}}{L} \bigg( 2 \sum_{ 0 < |j| \le m } \zeta^{jk} - (2m-1) \bigg) x + (-1)^{m+1} \Biggr).
\end{multline}
In this case, the coefficient of \(x^3\) is non-integer. That is made of three types of products: i) a product of one \(x^3\) and \(L^2-1\) constants; ii) a product of one \( x^2 \), one \(x\), and \(L^2-2\) constants; iii) a product of three \(x\)'s and \(L^2-3\) constants. It is easy to see that the coefficients of the sum of all terms of type i) and ii) is an integer. Note that
\[
    \sum_{k=0}^{L^2-1} \z^{jk} = 0
\]
unless \( j \) is a multiple of \( L^2 \). Then, for the sum of type i),
\begin{align}
    \sum_{k=0}^{L^2-1} \frac{1}{L} \biggl( - 2 (m-1) \sum_{ 0 < |j| \le m } \z^{jk} + m(2m-5) \biggr) = \; & \frac{1}{L} ( m(2m-5) L^2 ) \\
    = \; & m (2m-5) L
\end{align}
holds, and this is an integer.

Similarly, for the sum of type ii), we have
\begin{align}
    \frac{1}{L^2} &\sum_{ k \neq l } \biggl( 2 \sum_{ |j| \le m } \z^{jk} + 2m^2-3m-2 \biggr) \biggl( 2 \sum_{ |j| \le m } \z^{jl} - L \biggr).
\end{align}
This can be expanded as follows:
\[
    \frac{1}{L^2} \sum_{ k \neq l } \biggl( 4 \sum_{ |j_1| \le m } \z^{j_1k} \sum_{ |j_2| \le m } \z^{j_2l} -2L \sum_{ |j| \le m } \z^{jk} + 2 ( 2m^2-3m-2 ) \sum_{ |j| \le m } \z^{jl} - (2m^2-3m-2) L \biggr).
\]
Moreover, similarly to the M-type case, we know
\[
    \sum_{ k \neq l } \sum_{ |j_1| \le m } \z^{j_1k} \sum_{ |j_2| \le m } \z^{j_2l} = -(2m+1)L^2
\]
and
\[
    \sum_{ k \neq l } \sum_{ |j| \le m } \z^{jk} = \sum_{ k \neq l } \sum_{ |j| \le m } \z^{jl} = - L^2 (L^2-1).
\]
Therefore, the sum of type ii) is
\[
    - 4L + 2L(L^2-1) - 2(2m^2-3m-2)(L^2-1) - (2m^2-3m-2)L(L^2-1).
\]
This is an integer.

Finally, type iii):
\[
    \frac{1}{6} \frac{1}{L^3} \sum_{\substack{ \mathrm{dist}(k,l,s) }} \biggl( 2 \sum_{ |j| \le m } \z^{jk} - L \biggr) \biggl( 2 \sum_{ |j| \le m } \z^{jl} - L \biggr) \biggl( 2 \sum_{ |j| \le m } \z^{js} - L \biggr).
\]
Similarly to the M-type case, we know
\begin{multline}
    \sum_{\substack{ \mathrm{dist}(k,l,s) }} \sum_{ |j_1| \le m } \z^{j_1k} \sum_{ |j_2| \le m } \z^{j_2l} \sum_{ |j_3| \le m } \z^{j_3s}\\
    = L^2 ( L^2-1 ) ( L^2 - 2 ) -6m ( L^2-2 ) L^2 + 6m(m-1)L^2,
\end{multline}
\[
    \sum_{\substack{ \mathrm{dist}(k,l,s) }} \sum_{ |j_1| \le m } \z^{j_1k} \sum_{ |j_2| \le m } \z^{j_2l} = - L^3 (L^2-2),
\]
and
\[
    \sum_{\substack{ \mathrm{dist}(k,l,s) }} \sum_{ |j| \le m } \z^{jk} = - L^2 (L^2-1) (L^2-2).
\]
Consequently, the sum of type iii) is as follows:
\begin{align}
    & \frac{1}{6L^3} \Bigl( 8 \sum \sum \sum \sum - 12 L \sum \sum \sum + 3L^2 \sum \sum - L^3 \sum \Bigr)  \\
    = \; & \frac{1}{6L}\bigl( 8 \bigl( ( L^2-1 ) ( L^2 - 2 ) -6m ( L^2-2 ) + 6m(m-1) \bigr) \\
    & \qquad\qquad + 12 L^2 (L^2-2) - 3L^2(L^2-1)(L^2-2) - L^3 (L^2-1)(L^2-2) \bigr).
\end{align}
The latter three terms are multiples of \(L\), but the first term is not, as we checked in M type part.

Thus, the coefficient of \(x^3\) of \( f^{M,L}_{L^2}(x) \) is not an integer.
\end{proof}

\section*{Acknowledgment}
We would like to thank Rei Aoki, Rikako Teramura, and Natsuki Tsuka for their help with the calculations.

\end{document}